\documentclass[thmsa]{article}
\usepackage{amssymb}
\usepackage{amsmath,amsthm,amsfonts}  
\usepackage[applemac]{inputenc}  \usepackage{graphicx}
\pdfoutput=1
\usepackage{amssymb}
\usepackage{amsmath}
\usepackage{amsthm}
\usepackage{color}
\newtheorem{defi}{Definition}[section]
\newtheorem{prop}{Proposition}[section]
\newtheorem{lem}{Lemma}[section]
\newtheorem{theo}{Theorem}[section]

\newcommand{\Proof}{\noindent{\bf Proof.\,\,\,}}

\renewcommand{\qed}{\hfill $\square$}

\def\eqdef{\stackrel{\mbox{\tiny def}}{=}}     
\newcommand{\ket}[1]{|\kern.3ex#1\kern.3ex\rangle}
\newcommand{\bra}[1]{\langle\kern.3ex #1 \kern.3ex|}
\newcommand{\scalar}[2]{\langle\kern.3ex #1 \kern.3ex|\kern.3ex#2\kern.3ex\rangle}
\newcommand{\dt}{\mathfrak{t}}

\newcommand{\NN}{\mathcal{N}}

\def\C{\mathbb{C}}

\def\lg{\langle }
\def\rg{\rangle }
\def\deq{\stackrel{\mathrm{def}}{=}}

\definecolor{hervecolor}{rgb}{0.8,0,0.7}

\begin{document}

\title{Symmetric generalized binomial distributions}
\author{
H. Bergeron$^{\mathrm{a}}$, 
E.M.F.  Curado$^{\mathrm{b,c}}$,  \\
J.P. Gazeau$^{\mathrm{b,d}}$, 
 Ligia M.C.S. Rodrigues$^{\mathrm{b}}$(\footnote{e-mail: herve.bergeron@u-psud.fr, 
 evaldo@cbpf.br,
gazeau@apc.univ-paris7.fr, ligia@cbpf.br} )\\
\emph{$^{\mathrm{a}}$
Univ Paris-Sud, ISMO, UMR 8214, 91405 Orsay, France} \\
\emph{   $^{\mathrm{b}}$ Centro Brasileiro de Pesquisas Fisicas   } \\
\emph{   $^{\mathrm{c}}$ Instituto Nacional de Ci\^encia e Tecnologia - Sistemas Complexos}\\
\emph{  Rua Xavier Sigaud 150, 22290-180 - Rio de Janeiro, RJ, Brazil  } \\
\emph{  $^{\mathrm{d}}$ APC, UMR 7164,}\\
\emph{ Univ Paris  Diderot, Sorbonne Paris Cit\'e,  }  
\emph{75205 Paris, France} 
}

\maketitle
{\abstract{ 

In two recent articles we have examined a generalization of the binomial distribution associated with a sequence of positive numbers, involving asymmetric expressions of probabilities that break the symmetry {\it win-loss}. We present in this article another generalization (always associated with a sequence of positive numbers) that preserves the symmetry {\it win-loss}. This approach is also based on generating functions and presents constraints of non-negativeness, similar to those  encountered in our previous articles. }}

\tableofcontents

\section{Introduction}
\label{intro}

In two previous papers \cite{curadoetal2011, bergeronetal2012}, we have introduced Bernoulli-like (or binomial-like) distributions associated to  arbitrary sequences of positive numbers and we have given a set  of interesting properties.  These distributions break the {\it win-loss} symmetry of the ordinary binomial distribution. In general they are formal because they are not always positive. We have proved in \cite{curadoetal2011} that for a certain class of such sequences the corresponding  Bernoulli-like distributions are nonnegative and possess a Poisson-like limit law. We discussed also the question of non-negativeness in the general case and we examined under which conditions a probabilistic interpretation can be given to the proposed binomial-like distribution. In \cite{bergeronetal2012} we gave a different approach based on generating functions that allowed to avoid the previous difficulties: the constraints of non-negativeness were automatically fulfilled and a great number of analytical examples became available.

Our previous studies were mainly motivated by the construction and applications of non-linear coherent states (\cite{curadoetal2010}, see also \cite{dodonov2002,gazeaubook09} and references therein). In the present work, we give a different generalization of the binomial distribution that preserves the symmetry {\it win-loss} and which is more appropriate for the analysis of strongly correlated systems and their entropy behavior, as is done for instance in \cite{haneletal09}. 
In that paper the authors present a generalized stochastic model yielding $q$-exponentials for all ($q\neq 1$). (The $q$-exponential is a probability distribution, arising from the maximization of the Tsallis entropy under appropriate constraints, which is suitable for strongly correlated random variables \cite{tsallisBJP09}).  This is achieved by using the Laplace-de Finetti representation theorem for the Leibniz triangle rule \cite{laplace}, which embodies strict scale-invariance of interchangeable random variables. They  demonstrate that strict scale invariance mandates the associated extensive entropy to be Boltzmann-Gibbs (BG). By extensive we mean that the entropy is proportional to the size, defined in a certain sense, of the system. Our symmetric distributions do not in general fulfill the Leibniz triangle rule. Indeed, we prove that the only symmetric distributions, besides the ordinary binomial case,  which obey this rule are precisely those derived from the $q$-exponential. On the other side, we observe that in many other cases this rule is obeyed asymptotically as $n$ goes to infinity. Finally, in the case of the $q$-exponential we show that the corresponding BG entropy is extensive; in the cases the Leibniz rule is asymptotically verified, the BG entropy is asymptotically extensive.   

The questions of non-negativeness encountered in our previous works are also present in the symmetric  case and they are approached through a similar procedure. Interestingly,  the characterization of generating functions compatible with the constraints of non-negativeness leads to the set of solutions already found in \cite{bergeronetal2012}.

A potential application of our results concerns the notion of bipartite entanglement characterized by approaches involving the  Tsallis $q$-entropy \cite{kim10}. Another application concerns the construction of coherent states built from these symmetric deformed binomial distribution, in the same way spin coherent states are built from ordinary binomial distribution \cite{gazeaubook09}.

In Section \ref{genbernoulli} we review the asymmetric generalized binomial  distributions, the associated generating functions and some useful definitions from \cite{curadoetal2011, bergeronetal2012}.  We develop in Section \ref{sec:symmetricbin} the case of symmetric generalized binomial distributions with the necessary mathematical tools. They lead to a characterization of the solutions compatible with a complete probabilistic interpretation. The generalized Poisson distribution is introduced as well. In Section \ref{expval} we derive the general expression for expectation values and moments of generalized binomial distributions in the symmetric case.
In Section \ref{sec:examples} we present some analytical and numerical examples that illustrate the theoretical framework developed in Section \ref{sec:symmetricbin}. These examples belong to a class that we define as \textit{logarithmic scale invariant generating functions.} 

In Section \ref{sec:NLSCS} we present ``nonlinear spin'' coherent states built from those symmetric deformed binomial distributions, and suggest possible physical applications of these deformed distributions. In Section \ref{sec:conclusions} we present our final comments.  Technical details and relevant results with proofs are given in appendixes.

\section{The asymmetric Generalized Binomial Distribution: a brief review}
\label{genbernoulli}

For a sequence of $n$ independent trials with two possibles outcomes, ``win” and ``loss”, the probability of obtaining $k$ wins is given by the binomial  distribution $p_k^{(n)}(\eta)$
\begin{equation}
p_k^{(n)}(\eta) =\left(\begin{array}{c}n \\k \end{array}\right) \eta^k \,(1-\eta)^{n-k}=\dfrac{n!}{(n-k)!\, k!} \,\eta^k \,(1-\eta)^{n-k}
\end{equation}
where the parameter $0 \le \eta \le 1$ is the probability of having the outcome ``win'' and $1-\eta$, the outcome ``loss''.
Therefore, we can say that the Bernoulli binomial distribution above corresponds to the sequence of non-negative integers $n \in \mathbb{N}$.

In two recent articles \cite{curadoetal2011, bergeronetal2012} 
we generalized the binomial law\footnote{In our articles, the terms  \emph{generalized binomial} or \emph{binomial-like} or \emph{deformed binomial} distribution refer to mathematical expressions like (\ref{pgotico}) in which  $x_n$ are non-negative real arbitrary numbers. In most papers devoted to probability and statistics in which such  terms are used, their meaning is a lot more restrictive and  concerns counting combinatorics and integer numbers.} using an increasing sequence of nonnegative real numbers $\{ x_n \}_{n \in \mathbb{N}}$, where by convention $x_0=0$. The new ``formal'' (i.e. not necessarily positive) probabilities $\mathfrak{p}_k^{(n)}(\eta)$ are defined as
\begin{equation}
\label{pgotico}
\mathfrak{p}_k^{(n)}(\eta)=\dfrac{x_n!}{x_{n-k}! x_k!} \eta^k p_{n-k}(\eta)\equiv \binom{
x_n}{ x_k}\eta^k p_{n-k}(\eta)
\end{equation}
where the factorials $x_n!$ are given by
\begin{equation}
x_n!=x_1 x_2 \dots x_n, \quad x_0! \eqdef 1,
\end{equation}
and where the polynomials $p_n(\eta)$ are constrained by the normalization
\begin{equation}
\label{somapgotico}
\forall n \in \mathbb{N}, \quad \sum_{k=0}^n \mathfrak{p}_k^{(n)}(\eta)= 1.
\end{equation}
Note that this distribution is called asymmetric in the sense that the expression (\ref{pgotico}) is not invariant under  $k \to n-k$ and $\eta \to 1-\eta$, as it happens in the binomial case. 
The polynomials $p_n(\eta)$ can be obtained by recurrence from  Eq. (\ref{somapgotico}) , the first ones being $p_0(\eta)=1$ and $p_1(\eta)=1-\eta$. The $p_n=\mathfrak{p}_0^{(n)}$ possess a probabilistic interpretation as long as $\mathfrak{p}_k^{(n)}$ are themselves probabilities. This latter condition holds only if the $p_n(\eta)$ are nonnegative. Therefore non-negativeness of the $p_n$ for $\eta \in [0,1]$ is a key-point to preserve the statistical interpretation.
In this case,  the quantity $\mathfrak{p}_k^{(n)}(\eta)$  can still be interpreted as the probability of having $k$ wins and $n-k$ losses in a sequence of \emph{correlated} $n$ trials.  But, due to the asymmetry, there can be bias favoring either win or loss, even in the case $\eta = 1/2$. 

 In order to analyze this positiveness  constraint and characterize the solutions, an analysis based on generating functions has been developed in \cite{bergeronetal2012}.  

In \cite{curadoetal2011} 
it was proved that the analytical function 
\begin{equation}
\label{eqn:XnGFct}
\mathcal{N}(t)=\sum_{n=0}^\infty \dfrac{t^n}{x_n!}
\end{equation}
gives the generating function of  $p_k(\eta)$ (independently of the 
non-negativeness 
question) thanks to
\begin{equation}
\label{eqn:PnGFct}
\dfrac{\mathcal{N}(t)}{\mathcal{N}(\eta t)}=\sum_{k=0}^\infty p_k(\eta)\dfrac{t^k}{x_k!} 
\end{equation}
 we defined in \cite{bergeronetal2012} the set $\Sigma$ as
 the set of entire series $\mathcal{N}(t)=\sum_{n=0}^\infty a_n t^n$ possessing a non-vanishing radius of convergence and verifying $a_0=1$ and $\forall n \ge1, \, a_n>0$. We associate to each of these functions $\mathcal{N}$ the positive sequence $\{ x_n \}_{n \in \mathbb{N}}$ defined as $x_0=0$ and $\forall n \ge 1, \, x_n=a_{n-1}/a_n=n \, \mathcal{N}^{(n-1)}(0)/\mathcal{N}^{(n)}(0)$.
Then  we obtain $a_n=1/(x_n!)$ and  the definition of Eq.\eqref{eqn:XnGFct} is recovered.
We also defined $\Sigma_0$ as the set of entire series $f(z)=\sum_{n=0}^\infty a_n z^n$ possessing a non-vanishing radius of convergence and verifying the conditions $a_0=0$, $a_1 >0$ and $\forall n \ge 2, \, a_n \ge 0$. Finally, we proved that 
\begin{equation}
\label{sigma+}
\Sigma_+:=\left\{ \mathcal{N} \in \Sigma \, | \,  \forall \eta \in [0,1[, \, p_n(\eta) > 0 \, \right\}=\{ e^F \, | \, F \in \Sigma_0\}\, ; 
\end{equation}
$\Sigma_+$ is the set of generating functions solving the non-negativeness problem.\\

An interesting property reenforcing the probabilistic interpretation is the Poisson-like limit of the asymmetric deformed binomial distribution. If the sequence $ \{ x_n\}_{n\in \mathbb{N}}$ is such that, 
at fixed finite $m$, 
$\underset{n \to \infty}{\lim}{\frac{x_{n-m}}{x_n}} = 1$, then
the limit when  $n \to \infty$ of  $ \mathfrak{p}_k^{(n)}(\eta)$  with  $\eta = \dt/x_n$ is equal to
 $ \frac{1}{\mathcal{N}(\dt)}\, \frac{\dt^k}{x_k!}$.

\section{The Symmetric Generalized Binomial Distribution}
\label{sec:symmetricbin}

Using the notations of the previous section, we propose to study the following generalization of the binomial distribution,
\begin{equation}
\label{genbinsym}
\mathfrak{p}_k^{(n)}(\eta)=\dfrac{x_n!}{x_{n-k}! x_k!} q_k(\eta) q_{n-k}(1-\eta)\, ,
\end{equation}
where the $\{x_n\}$ is a non-negative sequence and the $q_k(\eta)$ are polynomials of  degree $k$, while $\eta$ is a running parameter on the interval $[0,1]$. The $\mathfrak{p}_k^{(n)}(\eta)$ are constrained by:
\begin{itemize}
\item[(a)]the normalization
\begin{equation}
\label{eq:polynorm}
\forall n \in \mathbb{N}, \quad \forall \eta \in [0,1], \quad \sum_{k=0}^n \mathfrak{p}_k^{(n)}(\eta)=1,
\end{equation}
\item[(b)]the non-negativeness condition (due to statistical interpretation)
\begin{equation}
\label{poscondsym}
\forall n,k \in \mathbb{N}, \quad \forall \eta \in [0,1], \quad \mathfrak{p}_k^{(n)}(\eta) \ge 0.
\end{equation}
\end{itemize}
The normalization condition \eqref{eq:polynorm} applied to the case $n=0$ leads to
\begin{equation}
\label{eq:non-negativeness}
\forall \eta \in [0,1], \quad \mathfrak{p}_0^{(0)}(\eta)=q_0(\eta) q_0(1-\eta)=1.
\end{equation}
Since $q_0$ is assumed to be a polynomial of degree $0$, we deduce that
$q_0(\eta)=\pm 1$.
\emph{In all the remainder we assume}
\begin{equation}
\label{q01} 
q_0(\eta)=1.
\end{equation}
Using this value of $q_0$ we deduce
\begin{equation}
\forall n \in \mathbb{N}, \quad \forall \eta \in [0,1], \quad \mathfrak{p}_0^{(n)}(\eta)=q_n(1-\eta).
\end{equation}
Therefore the non-negativeness condition \eqref{poscondsym} is equivalent to the non-negativeness of the polynomials $q_n$ on the interval $[0,1]$.
Here, like in the asymmetric case,  the quantity $\mathfrak{p}_k^{(n)}(\eta)$  can be interpreted as the probability of having $k$ wins and $n-k$ losses in a sequence of \emph{correlated} $n$ trials.  Besides, as we recover the  invariance under  $k \to n-k$ and $\eta \to 1-\eta$ of the binomial distribution, no bias (in the case $\eta = 1/2$) can exist favoring either win or loss.

\subsubsection*{Generating functions}

Setting aside the non negativeness condition, different generating functions (and different sets of polynomials) can be associated with the conditions expressed above. Indeed, given a sequence $\{x_n\}$ and its  generalized exponential $\mathcal{N}(t)$ introduced in \eqref{eqn:XnGFct}, we define the related generating function as
\begin{equation}
\label{genfctgen}
F(\eta; t) := \sum_{n=0}^\infty \dfrac{q_n(\eta)}{x_n!} t^n\, . 
\end{equation}
Due to the normalization property \eqref{eq:polynorm}, it is easy to show that
\begin{equation}
\label{FetaF}
F(\eta; t)\,F(1-\eta; t)= \mathcal{N}(t)\,, 
\end{equation}
which is a non-linear functional equation for the variable $\eta$ with a given boundary condition. The general solution is found from \cite{polyanin}:
\begin{equation}
\label{solgengen}
F(\eta; t) = \pm\sqrt{\mathcal{N}(t)} e^{\Phi(\eta,1-\eta;t)}\, , 
\end{equation}
where the arbitrary function $\Phi(x,y;t)$ is antisymmetric w.r.t. the permutation of variables $x$ and $y$, $\Phi(x,y;t) = -\Phi(y,x;t)$. 

\subsubsection*{The simplest class of generating functions}
 The simplest generating function giving rise to polynomials of binomial type \cite{umbral} corresponds to the choice $\Phi(x,y;t) = (x-y)\varphi(t)$, together with positive sign in \eqref{solgengen} and the boundary condition
 \begin{equation}
\label{boundCondG}
F(0; t)= 1 \Leftrightarrow F(1;t) = \mathcal{N}(t)\,. 
\end{equation}
This  leads to $\mathcal{N}(t) = e^{2\Phi(1,0;t)}= e^{2\varphi(t)}$ 
 and eventually to
 \begin{equation}
\label{simplePhi}
F(\eta; t) = e^{(2\eta)\varphi(t)} = (\mathcal{N}(t))^{\eta}\, . 
\end{equation}
 Hence, starting from Eq.\eqref{eqn:XnGFct} and an $x_n$-generating function $\mathcal{N}(t) \in \Sigma$ (the definition of $\Sigma$ was given above), let us define the functions $q_n$ as
\begin{equation}
\label{eq:GenFctqn}
\forall \eta \in [0,1], \quad \mathcal{N}(t)^\eta=\sum_{n=0}^\infty q_n(\eta)\dfrac{t^n}{x_n!} \, . 
\end{equation}
This definition makes sense since:
\begin{itemize}
\item[(a)] $\mathcal{N} \in \Sigma$ implies that $\mathcal{N}$ is analytical around $t=0$,
\item[(b)] $\mathcal{N} \in \Sigma$ implies $\mathcal{N}(0)=1$
\end{itemize}
Then $\ln (\mathcal{N}(t))$ is analytical around $t=0$, and therefore $\mathcal{N}(t)^\eta=\exp( \eta \ln (\mathcal{N}(t)))$ possesses a convergent series expansion around $t=0$ (for all $\eta \in \mathbb{C}$). 
The functions $q_n$ defined by \eqref{eq:GenFctqn} verify the following properties, whose proof is given in Appendix \ref{proof1}:
\begin{itemize}
\item[(a)] The two first polynomials $q_n(\eta)$ are $q_0(\eta)=1$ (according to our choice given by equation \eqref{q01}) and
\begin{equation}
q_1(\eta)=\eta ;
\end{equation}
more generally  the $q_n$ verify the recurrence relation
\begin{equation}
\label{eq:qnrecurrence}
\forall n \in \mathbb{N}, \, \forall \eta \in [0,1], \, q_{n+1}(\eta)= \eta \, \dfrac{x_{n+1}}{n+1} \, \sum_{k=0}^n \left(\begin{array}{c}x_n \\x_k \end{array}\right) \dfrac{n-k+1}{x_{n-k+1}}\,  q_k(\eta-1);
\end{equation}
\item[(b)] the $q_n$'s are polynomials of degree $n$ obeying
\begin{equation}
\forall n \in \mathbb{N}\, , \ q_n(1)=1, \, {\rm and } \ \forall \, n\ne 0\, , \ q_n(0)=0\, , 
\end{equation}
 and they fulfill the normalization condition \eqref{eq:polynorm};
\item[(c)] the $q_n$'s fulfill the functional relation
\begin{equation}
\label{eq:fctrelation}
\forall z_1, z_2 \in \mathbb{C}, \, \forall n \in \mathbb{N}, \, \ \sum_{k=0}^n \left(\begin{array}{c}x_n \\x_k \end{array}\right) q_k(z_1) q_{n-k}(z_2)=q_n(z_1+z_2)\, .
\end{equation}
\end{itemize}
We give in Appendix \ref{combqn} a combinatorial expression of the polynomials $q_n$ issued from the obvious formula 
\begin{equation}
\label{qndergen}
\left. q_n(\eta) = x_n!\frac{d^n}{dt^n} (\mathcal{N}(t))^{\eta} \right\vert_{t=0}\, . 
\end{equation}
We note that these polynomials, suitably normalized, are of binomial type \cite{umbral}. Indeed, with the definition 
$\tilde q_n(\eta) = \frac{n!}{x_n!} q_n(\eta)$ we have $\tilde q_n(z_1+z_2)= \sum_{k=0}^n \binom{n}{k} \tilde q_k(z_1) \tilde q_{n-k}(z_2)$.

We deduce from (b) that the series expansion in $t$ of the function $\mathcal{N}(t)^\eta$ generates a family of polynomials $q_n$ that fulfill one half of the sought conditions: it remains to specify under which conditions the non-negativeness problem can be solved.
Note that the usual binomial distribution is included in this family: it corresponds to the obvious choice $\mathcal{N}(t)=e^t$.

\subsubsection*{The non-negativeness problem}

The procedure is similar to the one developed in \cite{bergeronetal2012}. \\
Since we already know that $q_0(\eta)=1$ and $\forall n\ne 0, q_n(0)=0$, the non-negativeness condition is equivalent to specify that for any $\eta \in ]0,1], \, q_n(\eta)>0$ and then the function $t \mapsto \mathcal{N}(t)^\eta$ belongs to $\Sigma$. We prove in Appendix \ref{proof2} that the set $\tilde \Sigma_+$, defined as 
\begin{equation}
\tilde\Sigma_+:=\left\{ \mathcal{N} \in \Sigma \, | \,  \forall \eta \in ]0,1], \, q_n(\eta) > 0 \, \right\} \, ,
\end{equation}
is in fact
\begin{equation}
\label{tsigma+}
\tilde\Sigma_+=\{ e^F \, | \, F \in \Sigma_0\} = \Sigma_+\, .
\end{equation}
$\tilde\Sigma_+$ is the set of generating functions solving the non-negativeness problem for the symmetric case; it is actually the same $\Sigma_+$ as for the asymmetric case.

Equation  (\ref{tsigma+})   characterizes those  generating functions in the simplest class which  yield symmetric deformations of the binomial distribution preserving a statistical interpretation.
Therefore all coefficients $q_n^{(k)}(0)$ of the polynomials $q_n(\eta)$ are in fact positive, and then the $q_n$ are positive on $\mathbb{R}^+$.

Furthermore if $\mathcal{N}=e^F$ with $F(t)=\sum_{n=1}^\infty a_n t^n \in \Sigma_0$, the $x_n!$ and the $q_n$ can be generally expressed in terms of the complete Bell polynomials $B_n(a_1,a_2,\dots,a_n)$ \cite{comtet74} as
\begin{equation}
\label{xqbell}
x_n! = \dfrac{n!}{B_n(a_1,a_2,\dots,a_n)} \quad \textrm{and} \quad q_n(\eta) = \dfrac{B_n(\eta a_1, \eta a_2,\dots, \eta a_n)}{B_n(a_1,a_2,\dots,a_n)} \, ,
\end{equation}
thanks to the relation 
\begin{equation}
e^{F(t)} = \sum_{n=0}^\infty B_n(a_1,a_2,\dots,a_n)\frac{t^n}{n!} \,.
\end{equation}
Moreover, using the Bell polynomials properties \cite{comtet74} we also have from Eq. \eqref{xqbell}
\begin{equation}
\label{qpolybell}
q_0(\eta)=1 \quad \textrm{and} \quad \forall n \ge 1, \, q_n(\eta) = \sum_{k=1}^n \dfrac{B_{n,k}(a_1,a_2,\dots,a_{n-k+1})}{B_n(a_1,a_2,\dots, a_n)} \eta^k \, ,
\end{equation}
where the $B_{n,k}(a_1,a_2,\dots,a_{n-k+1})$ are the so-called partial Bell polynomials.\\
Starting from any sequence of \underline{positive} numbers $a_1, a_2, \dots$ such that $a_1 \ne 0$, Eqs \eqref{xqbell}-\eqref{qpolybell} give an extended set of explicit solutions of our  problem, letting aside all  convergence conditions for entire series.

\section{Expected values and moments}
\label{expval}
\subsection{Expectation value}
Using successively
\begin{equation}
\label{prodgensym}
\mathcal{N}(zt)^{\eta}\mathcal{N}(t)^{1-\eta} = \sum_{n=0}^{\infty}\frac{t^n}{x_n!}\left( \sum_{k=0}^n z^k  \mathfrak{p}_{k}^{(n)}(\eta)\right)\, , 
\end{equation}
\begin{equation}
\label{dergensym}
(z\partial_{z})^K\left\lbrack \mathcal{N}(zt)^{\eta}\mathcal{N}(t)^{1-\eta}\right\rbrack_{z=1} = \sum_{n=0}^{\infty}\frac{t^n}{x_n!}\langle k^K\rangle_n\, , 
\end{equation}
we find the simple expression for the expectation value of the $k$ variable:
\begin{equation}
\label{expeckK}
\langle k \rangle_n(\eta) = \eta n\, . 
\end{equation}
We note that we obtain the same value as for the ordinary binomial case. 
\subsection{Variance}
From \eqref{eqn:LnN} and \eqref{dergensym} we derive the formula for the second moment:
\begin{equation}
\label{secmometa}
\langle k^2 \rangle_n(\eta) = \eta n^2 + \eta(\eta-1) c_n\, , 
\end{equation}
with 
\begin{equation}
\label{defcn}
 c_n= \sum_{m=1}^{n-1} m(n-m)\,\binom{x_n}{x_m}\, q'_m(0)\,, \ n\geq 2\, \ c_0=0=c_1\, ,
\end{equation}
and
\begin{equation}
t^2 \dfrac{\mathcal{N}'(t)^2}{\mathcal{N}(t)}=\sum_{n=0}^\infty \dfrac{c_n}{x_n!} t^n \,.
\end{equation}
Assuming $\mathcal{N}(t)= \exp(\sum_{n=1}^\infty a_n t^n)$ and using \eqref{xqbell}, we also obtain, in terms of Bell polynomials, that
\begin{equation}
c_n = \dfrac{n!}{B_n(a_1,a_2,\dots,a_n)} \sum_{m=1}^{n-1} \dfrac{m}{(n-m-1)!} a_m B_{n-m}(a_1,a_2,\dots,a_{n-m}) \,.
\end{equation}

Thus the variance is given by
\begin{equation}
\label{varianeta}
(\sigma_k)^2_n (\eta) = \eta(1-\eta) (n^2 - c_n)\, , 
\end{equation}
which entails the interesting property $c_n\leq n^2$ for all $n$. 

\section{Examples with logarithmic scale invariance}
\label{sec:examples}
\subsection{Logarithmic scale invariance}
\label{sec:logsi}
Suppose that the sequence $\{x_n\}$ depends  on a parameter $\alpha$, $x_n\equiv x_n(\alpha)$, in such a way that for the primary generating function
\begin{equation*}
\mathcal{N}(t) \equiv \mathcal{N}(t,\alpha)\, ,  
\end{equation*}
we have
\begin{equation}
\label{logscinv1}
(\mathcal{N}(t,\alpha))^{\eta}= \mathcal{N}(\eta t,\eta\alpha)\, , 
\end{equation}
i.e., its logarithm is homogeneous of degree 1, $ \ln (\mathcal{N}(\eta t,\eta\alpha))= \eta \ln (\mathcal{N}(t,\alpha))$. Choosing $\eta = 1/\alpha$ allows us to write also
\begin{equation}
\label{Neta}
\mathcal{N}(t,\alpha) = \mathcal{N}(t/\alpha,1).
\end{equation}
The logarithmic scale invariance allows to give by comparison the interesting  expression for polynomial $q_n$:
\begin{equation}
\label{expscqn}
q_n(\eta,\alpha)= \eta^n\frac{x_n(\alpha)!}{x_n(\eta\alpha)!}\,. 
\end{equation} 
For the sake of more simplicity in the notation, from now on we will write $q_n(\eta)$ instead of $q_n(\eta,\alpha)$.

This property is shared by the $3$ examples studied below. It is interesting to remark that the entropy, which is also an homogeneous function of degree one of its natural variables, can be written in the microcanonical ensemble as $S = \ln \Omega$, where $ \Omega$ is the number of microscopic states of the system. This suggests that one could interpret $\mathcal N(t)$ as the number of states of the system, a number which is modified by the presence of correlations among the events.  

\subsection{First example: ``$q$-exponential''}
\label{sec:example1}
\begin{equation}
\label{genex1}
\NN (t) = \left( 1- \frac{t}{\alpha}\right)^{-\alpha}\, , \quad \alpha > 0\, . 
\end{equation}
We first note that if $\alpha \to \infty$ then $\NN(t) \to e^t$, i.e. we return to the ordinary binomial case. The corresponding sequence is bounded by $\alpha$ and  given by
\begin{equation}
\label{xnex1}
x_n = \frac{n\alpha}{n + \alpha -1}\, , \quad \lim_{n\to \infty}x_n = \alpha \, . 
\end{equation}
For the factorial we have:
\begin{equation}
\label{factxn1}
x_n! = \alpha^n \frac{\Gamma(\alpha) n!}{\Gamma(n+\alpha)}=  \frac{ \alpha^n n!}{(\alpha)_n} = \frac{\alpha^n}{\binom{n+\alpha -1}{n}}\, ,
\end{equation}
where $(z)_n = \Gamma(z+n)/\Gamma(z)$ is the Pochhammer symbol.
The corresponding polynomials are   given by
\begin{equation}
\label{xnex1}
q_n(\eta) =  \frac{\Gamma (\alpha)}{\Gamma (n+ \alpha)}\, \frac{\Gamma (n+ \alpha \eta)}{\Gamma (\alpha \eta)}=  \frac{ (\alpha\eta)_n}{ (\alpha)_n}\, , 
\end{equation}
and satisfy the recurrence relation
\begin{equation}
\label{recrelex1}
q_n(\eta) =  \frac{ n + \alpha \eta -1}{n + \alpha -1}\, q_{n-1}(\eta)\, , \quad \mbox{with}\ q_0(\eta) = 1\, . 
\end{equation}
In particular $q_1(\eta) = \eta$. 
The  distribution $\mathfrak{p}_k^{(n)}(\eta)$ defined by these polynomials is given by
\begin{align}
\label{genbinsymA}
\mathfrak{p}_k^{(n)}(\eta)= &\binom{n}{k}\, \frac{\Gamma(\alpha)}{\Gamma(\alpha +n)}\, \frac{\Gamma(\eta\alpha +k)}{\Gamma(\eta\alpha)}\,\frac{\Gamma((1-\eta)\alpha +n-k)}{\Gamma((1-\eta)\alpha)}\\
\label{genbinsymB} =&\frac{\binom{-\eta \alpha}{k}\, \binom{-(1-\eta) \alpha}{n-k}}{\binom{- \alpha}{n}}\, . 
\end{align}
\paragraph{Comment}

This distribution reminds us of the hypergeometric discrete distribution  \cite{rice07}.
We recall that the  hypergeometric distribution  describes the probability of $k$ successes in $n$ draws from a finite population of size $N$ containing $m$ successes without replacement (e.g. lotto or bingo game):
\begin{equation}
\label{hypergeodist}
\mathfrak{h}_k^{(n)}= \frac{\binom{m}{k}\, \binom{N-m}{n-k}}{\binom{N}{n}}\, .
\end{equation}
Normalization results from the famous combinatorial identity
\begin{equation}
\label{combident}
\sum_{k=0}^m \binom{m}{k}\, \binom{N-m}{n-k}= \binom{N}{n}\, .
\end{equation}
Such an identity can be extended to real or complex numbers. Here, we can replace $N \rightarrow -\alpha$ and $m\rightarrow -\eta\alpha$ and so $\eta = m/N$.  \\

Hence, the polynomials $q_n(\eta)$ corresponding to the hypergeometric distribution read in terms of $N$ and $m$ as:
\begin{equation}
\label{qnhypergeom}
q_n(\eta) = q_n\left(\frac{m}{N}\right)= \frac{(N-n)!\,m!}{N!\,(m-n)!}= \frac{\binom{m}{n}}{\binom{N}{n}}\, . 
\end{equation}
We notice that they cancel when $n> m$. The sequence $(x_n)$ and the corresponding factorials are given by
\begin{equation}
\label{hypergxn}
x_n= \frac{Nn}{N-n+1}\, , \quad x_n! = \frac{N^n}{\binom{N}{n}}\, . 
\end{equation}

As indicated in Eq. \eqref{genex1}, in the present case we need $\alpha>0$  to agree with our definition of $\Sigma_+$. Since $\Sigma_+$ contains all solutions that give positive polynomials on the interval $[0,1]$, the fact that \emph{the substitution ``$N \rightarrow -\alpha$ and $m\rightarrow -\eta\alpha$ and so $\eta = m/N$" gives well-defined probabilities} seems to be contradictory. In fact it is not, because only special values of $\eta$ are chosen. 

Returning to the distribution \eqref{genbinsymA}, the expression of $c_n$ is found to be
\begin{equation}
\label{qgausscn}
c_n=  n(n-1)\frac{\alpha}{1+\alpha}\,.
\end{equation}
Therefore, besides the mean value $\langle k \rangle_n = \eta n$,  we have the simple expression for  the variance:
\begin{equation}
\label{varsymex1}
(\sigma_k)^2_n(\eta) = n^2 \eta(1-\eta)\frac{1 + \alpha/n}{1 + \alpha}\, . 
\end{equation}
As expected, we check that at large $\alpha$ the behavior of $\sigma_k$  is ordinary binomial. We also note that it becomes proportional to  the mean value at large values of $n$.  

\subsubsection*{Leibniz triangle rule}

We prove in Appendix \ref{leibqgauss} that the  $q$-exponential is the unique case for which the  Leibniz triangle rule holds  exactly true. 
Precisely, the relation
\begin{equation}
\label{leibnitz1}
\varpi_k^{n-1} = \varpi_k^{n} + \varpi_{k+1}^{n} 
\end{equation}
is  verified by the  quantities, and only by them,
\begin{equation}
\label{tildp1}
\varpi_k^{n}(\eta) := \frac{\mathfrak{p}_k^{(n)}(\eta)}{\binom{n}{k}}= \binom{y_n}{y_k} q_k(\eta) q_{n-k}(1-\eta)=  \frac{(\eta\alpha)_k\, ((1-\eta)\,  \alpha)_{n-k}}{(\alpha)_n}\, , 
\end{equation}
with $ y_n:= x_n/n$. These relations generalize the $r_k^n$'s of \cite{haneletal09}. The latter correspond to the limit case $\alpha \to \infty$.

\subsubsection*{Entropy}
Following \cite{haneletal09}, the above  result (\ref{leibnitz1})  implies that the entropy which is extensive when $n\to \infty$ is, like in the ordinary binomial case,  Boltzmann-Gibbs.  Let us explore in a more comprehensive numerical and analytical way this statistical aspect of the considered distribution. A natural definition of the BG entropy for the distribution (\ref{genbinsymA}) is 
\begin{equation}
\label{BGentr1}
S^{(n)}_{\mathrm{BG}} = - \sum_{k=0}^{n} \binom{n}{k} \varpi_k^{n}(\eta)\ln\varpi_k^{n}(\eta)\, .
\end{equation}
The presence of the binomial coefficient in the sum (\ref{BGentr1}) means a counting of the degeneracy. Its extensive property is shown in Figure \ref{entropy} where it is compared with  the Tsallis entropy  $S^{(n)}_q \deq (1- \sum_{k=0}^{n}\binom{n}{k} (\omega_k^n)^q)/(q-1)$. 
\begin{figure}[htb!]
\begin{center}
\includegraphics[width=2.6in]{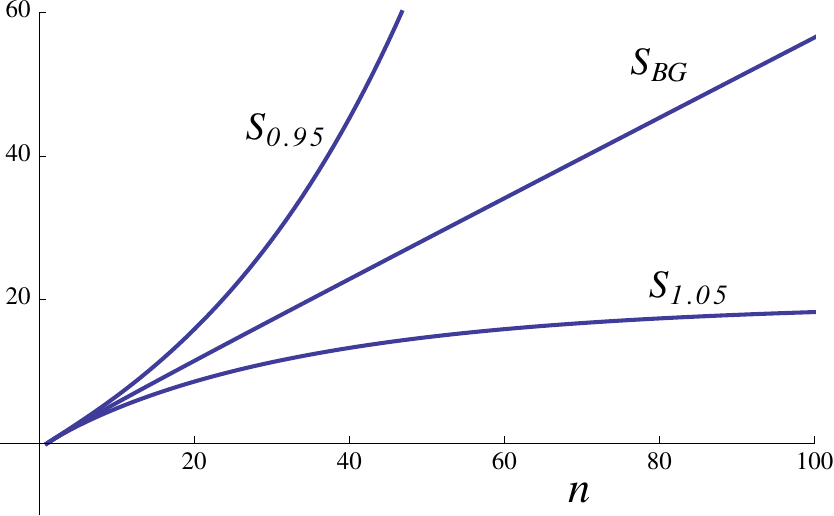}
\caption{$n$-dependence of entropies for the  ``$q$-exponential'' case 
($\mathcal{N}(t) = (1-t/\alpha)^{-\alpha}$).  The Boltzmann-Gibbs (BG) and Tsallis ($S_q$) entropies of 
the distribution are shown  for $\eta = 1/2$ and $\alpha = 3$. Upper curve is for $q=0.95$. Bottom curve is for $q=1.05$.
}
\label{entropy}
\end{center}
\end{figure}

\newpage

\subsubsection*{Limit distribution}

It is interesting to check (numerically) that the limit distribution of \eqref{genbinsymA} as $n\to \infty$ is of Wigner type, 
 as is shown in Figure \ref{limdistfigqg1}. We remark that the Wigner distribution is a kind of $q$-Gaussian, with $q=-1$. 
 There are two curves in this graph, one for $n=1000$ and the other for $n=2000$, showing that the limit distribution is already reached.
 
 \begin{figure}[htb!]
\begin{center}
\includegraphics[width=2in]{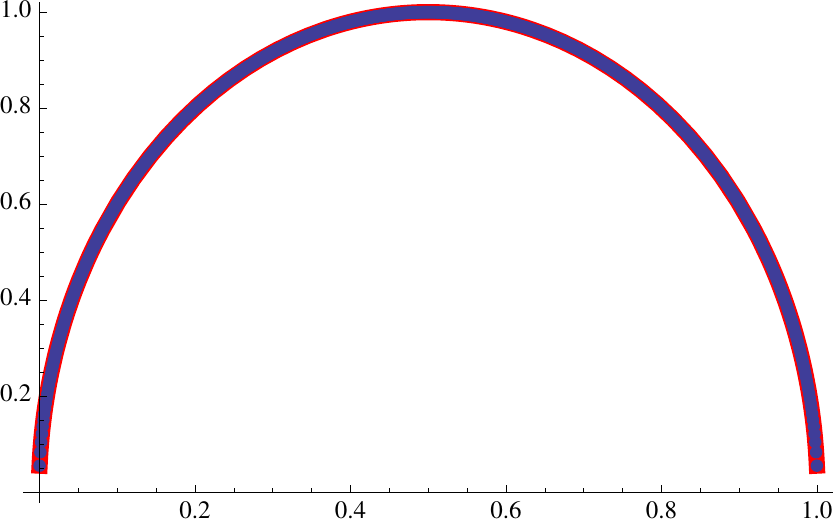}
\caption{
Limit distribution of the distribution \eqref{genbinsymA}.}
\label{limdistfigqg1}
\end{center}
\end{figure}

\subsection{Second example: modified Abel polynomials}
\label{sec:example2}
\begin{equation}
\label{genex2}
\NN (t) =e^{-\alpha W(-t/\alpha)}\, , \quad \alpha > 0\, , 
\end{equation}
where $W$ is the Lambert function, i.e.  solving  the functional equation $W(t)e^{W(t)} = t$. 
We first note that if $\alpha \to \infty$ then $\NN(t) \to e^t$. The corresponding sequence is bounded by $\alpha/e$ and  given by
\begin{equation}
\label{xnex2}
x_n = \frac{n\alpha}{n + \alpha }\left(1- \frac{1}{n + \alpha}\right)^{n-2}\, , \quad \lim_{n\to \infty}x_n = \alpha/e \, . 
\end{equation}
We also note that $x_n \to n$ as $\alpha \to \infty$.  The corresponding factorial is
\begin{equation}
\label{factabel}
x_n!= n!\, \frac{\alpha^{n-1}}{(n+\alpha)^{n-1}}\, .
\end{equation}
The  polynomials $q_n$'s are   given by
\begin{equation}
\label{xnex2}
q_n(\eta) =  \eta \frac{\left(\eta + \frac{n}{\alpha}\right)^{n-1}}{\left(1 + \frac{n}{\alpha}\right)^{n-1}}\, . 
\end{equation}
We verify that  $q_0(\eta) = 1$ and $q_1(\eta) = \eta$. 

The corresponding probability distribution is given by:
\begin{equation}
\label{abelprobdist}
\mathfrak{p}_k^{(n)}(\eta) = \binom{n}{k}\eta(1-\eta) \frac{(\eta + k/\alpha)^{k-1}(1-\eta + (n-k)/\alpha)^{n-k-1}}{(1+n/\alpha)^{n-1}}\,. 
\end{equation}
In order to calculate the variance, we need the explicit form of $c_n$ defined in \eqref{defcn}. We first note that 
\begin{equation*}
q'_n(0) = \frac{n^{n-1}}{(n+\alpha)^{n-1}}\,.
\end{equation*}
There results for  $c_n$:
\begin{align}
\label{abelcn}
\nonumber c_n&= n! \frac{\alpha}{(n+\alpha)^{n-1}}\sum_{m=0}^{n-2} \frac{(m+1)^{m}}{m!}\,\frac{(n-1+\alpha -m)^{n-2-m}}{(n-2 -m)!}\\
&= n(n-1)\frac{\alpha}{(n+\alpha)^{n-1}}\mathcal{T}_{n-2}(\alpha+2)\, , 
\end{align}
where the polynomials $\mathcal{T}_{n}(z)$ are defined by the generating function
\begin{equation}
\dfrac{1}{(1+W(x))^2} e^{- z W(x)} = \sum_{n=0}^\infty \dfrac{\mathcal{T}_{n}(z)}{n!} (-x)^n \, .
\end{equation} 
Interesting  properties of the $\mathcal{T}_{n}(z)$ are given in \cite{Tnz}, particularly related to  combinatorics (e.g. subfactorials or derangement numbers). We prove in \cite{Tnz} its relation to 
 the incomplete Gamma function \cite{abramovitz}:
 \begin{equation}
\label{igamnz}
\mathcal{T}_n(z) = e^{z+n}\, \Gamma(n+1,z+n) \,,
\end{equation}
where
 \begin{equation}
\label{incgamfct}
\Gamma(a,x) = \int_{x}^{\infty}t^{a-1}\,e^{-t}\,dt\, , \quad \mathrm{Re}(a) >0\, , 
\end{equation}
The expression for the variance is then given by 
\begin{align}
\label{varsymex2}
\nonumber (\sigma_k)^2_n(\eta) &= n^2\eta(1-\eta)\left(1- \frac{n-1}{n}\frac{\alpha}{(n+\alpha)^{n-1}}\mathcal{T}_{n-2}(\alpha+2)\right)\\
&= n^2\eta(1-\eta)\left(1- \frac{n-1}{n}\frac{\alpha}{(n+\alpha)^{n-1}}\sum_{k=0}^{n-2}\frac{(n-2)!}{k!}(n +\alpha -2)^k\right)\, . 
\end{align}
From \eqref{z+n}  we check that at large $\alpha$  the behavior of $\sigma_k$  is ordinary binomial. We also note that it becomes proportional to  the mean value at large values of $n$.  

In this case, we note that the Leibniz triangle is satisfied asymptotically as $n \to \infty$.

\subsection{Third example: Hermite polynomials}
\label{sec:example3}

\begin{equation}
\label{genex2}
\NN (t) =e^{t+ \frac{a}{2}t^2}\, , \quad 0< a <1\, . 
\end{equation}
The logarithmic scale invariance holds with $\alpha:= 1/a$. 
The corresponding sequence ${x_n}$ has the same factorial as in  the asymmetric case \cite{bergeronetal2012}:
\begin{equation}
\label{xnfat2} 
x_n! = \left[  \frac{i^n \left(\frac{a}{2}\right)^{n/2}}{n!} H_n \left( 
\frac{-i}{\sqrt{2 a}}
\right) \right]^{-1} = \left\lbrack\sum_{m=0}^{\left\lfloor\frac{n}{2}\right\rfloor} \frac{(a/2)^m}{m!(n-2m)!} \right\rbrack^{-1}:= \frac{1}{\varphi_{n}(a)}\,.
\end{equation}
In particular, $x_1!=x_1= 1$, $x_2! = 2/(a+1)$. Also, $x_n = \varphi_{n-1}(a)/\varphi_{n}(a)$, and  
we  know from \cite{bergeronetal2012}  that $x_n \approx \sqrt{n/a}$ as $n \to \infty$.  
The corresponding polynomials and probability distributions are respectively  given by
\begin{equation}
\label{xnex2}
q_n(\eta) =  \frac{x_n!}{n!} \left(i\sqrt{\frac{a \eta}{2}}\right)^n\, H_n\left(-i\sqrt{\frac{\eta}{2a}} \right)\,  
\end{equation}

and
\begin{equation}
\label{hermprobdist}
\mathfrak{p}_k^{(n)}(\eta)= \eta^k (1-\eta)^{n-k} \frac{\varphi_k(a/\eta)\varphi_{n-k}(a/(1-\eta))}{\varphi_n(a)}\,. 
\end{equation}
 
We know from the general case that $q'_0=0$ and $q'_1(0)=1$. The expression of  $q'_2(0)$ is found to be 
\begin{equation*}
q'_2(0) =\frac{a}{1+a}\,,
\end{equation*}
and we have $q'_n(0) = 0$ for all $n\geq 3$. Therefore, we have for the $c_n$'s, $n\geq 2$,  the following  expressions: 
\begin{align}
\label{hermitecn}
 c_2&= \frac{2}{1+a}\, , \\
\nonumber c_n &= n(n-1)\left\lbrack 1 + 2 \frac{H_{n-2}\left( \frac{-i}{\sqrt{2 a}} \right)}{H_{n}\left(\frac{-i}{\sqrt{2 a}} \right)} \right\rbrack\\ 
&=n(n-1)\left\lbrack 1  - \frac{a}{n(n-1)}\frac{\varphi_{n-2}(a)}{\varphi_{n}(a)} \right\rbrack \, , \quad n\geq 3\,.
\end{align}

The expression for the variance is then given by 
\begin{equation}
\label{varsymex2}
 (\sigma_k)^2_n(\eta) =  n\eta(1-\eta)\left\lbrack 1 + \frac{a}{n}\frac{\varphi_{n-2}(a)}{\varphi_n(a)}\right\rbrack\, .
\end{equation}
In this case, we check trivially  that for $a=0$  the behavior of $\sigma_k$  is ordinary binomial. On the other hand,  at large $n$, 
\begin{equation*}
\frac{\varphi_{n-2}(a)}{\varphi_n(a)} = \frac{\varphi_{n-2}(a)}{\varphi_{n-1}(a)}\frac{\varphi_{n-1}(a)}{\varphi_n(a)} \sim \frac{n}{a}\, , 
\end{equation*}
so the behavior of $\sigma_k$  is ``quasi'' ordinary binomial at large $n$:
\begin{equation}
\label{hermasvar}
(\sigma_k)_n(\eta) \sim \sqrt{2n \eta(1-\eta)}\, . 
\end{equation}

Also in this example the Leibniz triangle rule is obeyed only asymptotically as $n \to \infty$.

\section{Nonlinear coherent states}
\label{sec:NLSCS}
{We introduce here  families of coherent states built from the polynomials $q_n$, following the construction procedure explained in \cite{gazeaubook09}.  
The basic idea of this section is view the polynomials $q_n(\eta)$ as \emph{deformations} of $\eta^n$, associated to a deformation of the usual factorial $n!$. Although it seems natural to consider the $x_n!$ as a good candidate for the $n!$-deformation, this choice does not allow in general to obtain the so-called resolution of the identity. To avoid this problem we need in fact to introduce a new deformation of integers based on integrals.\\
In the following we restrict ourselves to the subset $\Sigma_+^{(1)} \subset \Sigma_+$ of generating functions $\Sigma_+^{(1)}=\{e^{\sum_{n=1}^\infty a_n t^n} | a_1=1, a_n \ge 0, \sum_{n=1}^\infty a_n < \infty\}$. We do not miss anything by using $\Sigma_+^{(1)}$. Indeed, starting from any function $F(t)=\sum_{n=1}^\infty a_n t^n \in \Sigma_0$ with a radius of convergence $R$, the function $\mathcal{N}=e^{{F_0}}$ with ${F_0}(t)=\alpha F(t/\alpha)/a_1$ and $\alpha>R$ belongs to $\Sigma_+^{(1)} $. Now if $\mathcal{N} \in \Sigma_+^{(1)}$, the first element of the sequence $\{x_n\}$ is $x_1=1$, and we know that the value $\mathcal{N}(1)$ is well-defined.\\
Starting from a generating function $\mathcal{N}=e^F \in \Sigma_+^{(1)}$ and the corresponding polynomials $q_n$, we define
\begin{equation}
\label{deffandb}
f_n = \int_0^\infty q_n(\eta) e^{-\eta} d\eta \quad \textrm{and} \quad b_{m,n}=\int_0^1 q_m(\eta) q_n(1-\eta) d\eta \,.
\end{equation}
The $f_n$ and $b_{m,n}$ are  deformations of the usual factorial and beta function deduced from their usual integral definitions by the substitution $\eta^n \mapsto q_n(\eta)$. We prove in Appendix \ref{proofsigmaprop} (for $\mathcal{N} \in \Sigma_+^{(1)}$) the following properties:
\begin{equation}
\label{fbproperties}
x_n! \le f_n \, , \forall \eta \in \mathbb{R}^+, \, \sum_{n=0}^\infty \dfrac{q_n(\eta)}{f_n} < \infty \quad \textrm{and} \quad b_{m,n} \ge \dfrac{x_m! x_n!}{(m+n+1)!} \,.
\end{equation}
We introduce the function $N(z)$ defined on $\mathbb{C}$ as
\begin{equation}
\forall z \in \mathbb{C}, \, N(z) = \sum_{n=0}^\infty \dfrac{q_n(z)}{f_n} \,.
\end{equation}
This definition makes sense since from Eq. \eqref{fbproperties}
\begin{equation}
\sum_{n=0}^\infty \left| \dfrac{q_n(z)}{f_n} \right| \le \sum_{n=0}^\infty \dfrac{q_n(|z|)}{f_n} < \infty .
\end{equation}

\subsection{Nonlinear coherent states on the complex plane}
Let $\mathcal{H}$ be some separable Hilbert space with orthonormal basis $\{ | e_n \rangle \}, n \in \mathbb{N}$. We define the (normalized) coherent state $\{ |z \rangle \}$, $z \in \mathbb{C}$, as
\begin{equation}
| z \rangle = \dfrac{1}{\sqrt{N(|z|^2)}} \sum_{n=0}^\infty \dfrac{1}{\sqrt{f_n}} \sqrt{q_n(|z|^2)}\, e^{i \, n \arg(z)} |e_n \rangle \,.
\end{equation}
These states verify the following resolution of the unity $1_d$ of $\mathcal{H}$:
\begin{equation}
\int_{\mathbb{C}} \dfrac{d^2z}{\pi} e^{- |z|^2} N(|z|^2) \, |z \rangle \langle z | = 1_d \, .
\end{equation}
These coherent states are a natural generalization of the usual harmonic coherent states that correspond to the special polynomials $q_n(\eta)= \eta^n$. The $q_n(\eta)=\eta^n$ are associated to the generating function $\mathcal{N}(t)=e^t \in \Sigma_+^{(1)}$ that gives the usual binomial distribution.

{\subsection{Nonlinear spin coherent states}

These states can be considered as generalizing the  spin coherent states $| \theta,\phi\rg$  in the $2j+1$-dimensional space Hilbert space $\mathcal{H}_j \sim \C^{2j+1}$ with $2j=n$.  To each point in the  unit  sphere $S^2$ with spherical coordinates $(\theta, \phi)$, $0\leq \theta \leq \pi$, $0\leq \phi< 2\pi$ there corresponds the unit norm state 
\begin{equation}
\label{genspinCS}
| \theta,\phi\rg= \dfrac{1}{\sqrt{\varpi(\theta)}} \sum_{m=-j}^{j}\sqrt{\frac{q_{j-m} \left( \cos^2 \theta/2  \right) q_{j+m} \left( \sin^2 \theta/2 \right) }{b_{j-m,j+m}}} e^{im\phi}|j,m\rg\, , 
\end{equation}
where the $b_{m,n}$ are defined in Eq. \eqref{deffandb} and $\varpi(\theta)$ is given by
\begin{equation}
\varpi(\theta) = \sum_{m=-j}^{j}\frac{q_{j-m} \left( \cos^2 \theta/2  \right) q_{j+m} \left( \sin^2 \theta/2 \right) }{b_{j-m,j+m}}\,.
\end{equation}
Here $\{ |j,m\rg\, , \, -j\leq m \leq j\}$ is the familiar quantum angular momentum orthonormal basis of $\mathcal{H}_j $ (actually it could be \underline{any} orthonormal basis). Now  the family of states (\ref{genspinCS}) resolves the unity $1_d$  in $\mathcal{H}_j $:
\begin{equation}
\label{resunitj}
\frac{1}{4\pi}\int_{S^2} \sin\theta\, d\theta\, d\phi \, \varpi (\theta) \, | \theta,\phi\rg\lg  \theta,\phi| = 1_d\, . 
\end{equation}
These states can be named coherent states (in a wide sense) which are deformations of the spin coherent states.

\section{Conclusion}
\label{sec:conclusions}

In previous articles we have presented a method to construct deformed binomial and Poisson distributions which can be interpreted as probabilities as well. There, the deformed distributions did not preserve the symmetry between {\it win} and {\it loss}. In this paper we studied a normalized {\it win-loss} symmetrical generalization of the binomial distribution that fulfills the non-negative condition and can be interpreted as the probability associated to a sequence of correlated trials.  The related generating functions were defined in a general way and the non-negativeness problem was solved for the simplest class.  Three examples were presented; they are interesting at various levels, not the least  being that they lead to manageable results. Also, they pertain to a class of what we defined as \textit{logarithmic scale invariant generating functions}. The logarithmic scale invariance  is a particularly interesting property as it points to an interpretation of  the primary generating function as the number of states of a system where correlations are present. In one of the presented examples, namely the $q$-exponential, the Leibniz triangle rule is exactly obeyed and therefore the extensive entropy is Boltzmann-Gibbs. We proved that this is the only case and this result itself is new and rather interesting. In many other cases, the Leibniz rule is only asymptotically true. In a forthcoming paper we will continue a systematic study of the entropy behavior for various types of deformations of the binomial distribution. 

We also introduced a new class of coherent states built from the symmetric deformed distributions. In future works we think of using these coherent states to quantize the complex plane of the unit sphere in the spirit of \cite{gazeaubook09}. 

Finally, we note that there is another possible generalization of the binomial distribution along the same lines, in which  the probabilities of {\it win} and {\it loss} are given, respectively by two different polynomials, according to

\begin{equation}
\label{asymetric2pol}
\mathfrak{p}_k^{(n)}(\eta)=\dfrac{x_n!}{x_{n-k}! x_k!} p_k(\eta) q_{n-k}(1-\eta)\, .
\end{equation}

\section*{Acknowledgments}
H. Bergeron and J.P. Gazeau thanks the CBPF and the CNPq for financial support and CBPF for hospitality. E.M.F. Curado acknowledges CNPq and FAPERJ for financial support. 

\appendix 
\section{Proof of the properties of $q_n$}
\label{proof1}
Let us define $G(t,\eta)=\mathcal{N}(t)^\eta$.\\
(a) We remark first that $G(0,\eta)=1$ (because $\mathcal{N}(0)=1$) then $q_0(\eta)=1$. Moreover $\dfrac{\partial G}{\partial t}(0,\eta)= \eta \mathcal{N}'(0) G(0,\eta-1)=\eta/x_1$, therefore $q_1(\eta)=\eta$. \\
More generally we have $\dfrac{\partial G}{\partial t}(t,\eta)= \eta \mathcal{N}'(t) G(t,\eta-1)$. Identifying the series expansion in $t$ of the left and right hand side, we obtain the sought recurrence relation \eqref{eq:qnrecurrence}.\\
(b) Knowing $q_0(\eta)=1$ and the recurrence relation \eqref{eq:qnrecurrence}, we deduce that the functions $q_n$ are polynomials, the degree of $q_n$ being $n$. Furthermore the recurrence relation shows that $\forall n\ne 0, q_n(0)=0$ and the series expansion in $t$ of $G(t,1)$ shows that $\forall n \in \mathbb{N}, q_n(1)=1$. Finally we remark that $G(t,\eta) G(t,1-\eta)=\mathcal{N}(t)$. Using the series expansion in $t$ of the left hand side of this equation, we obtain the normalization condition \eqref{eq:polynorm}.\\
(c) More generally since the series expansion in $t$ of $G(t,\eta)$ is valid for any $\eta \in \mathbb{C}$ and since $G(t,z_1) G(t,z_2)=G(t,z_1+z_2)$, we obtain the functional relation \eqref{eq:fctrelation}.

\section{Combinatorial expression of $q_n$}
\label{combqn}
In the case where the generating function of polynomials $q_n$ is $(\mathcal{N}(t))^{\eta}$ the application of 0.430-2 in \cite{gradstein} to \eqref{qndergen} leads to the complete (although not so easily manageable) expression for polynomials $q_n$, for $n\geq 1$:
\begin{equation}
\label{}
q_n(\eta)= x_n! \sum_{m=1}^{n}\frac{\Gamma(\eta+1)}{\Gamma(\eta -m +1)}\sum_{\mathcal{I}(i_s)}\frac{n!}{i_1!i_2!\dots i_k!}\frac{1}{\prod_{s=1}^{k}(x_{s}!)^{i_s}}\, ,
\end{equation}
where the multiple summation symbol $\mathcal{I}(i_s)$ means summation on all possible variable-length -uples of nonnegative  integers $i_1$, $i_2$, $\dotsc$, $i_k$,  submitted to the 2  constraints
\begin{align*}
&\sum_{s=1}^{k} si_s=n\, , \\
&\sum_{s=1}^{k} i_s=m\, .
\end{align*}

\section{Proof of the non-negativeness of $q_n$}
\label{proof2}

\begin{defi} Using the notation $G_{\mathcal{N},\eta}(t)=\mathcal{N}(t)^\eta$, $\Sigma_+$ is defined as
\begin{equation}
\label{def:sigmaplus}
\tilde{\Sigma}_+=\left\{ \mathcal{N} \in \Sigma \, | \,  \forall \eta \in ]0,1], \, G_{\mathcal{N},\eta} \in \Sigma \, \right\}
\end{equation}
\end{defi}

\begin{prop}
\label{prop:sigmaplus}
The set $\tilde{\Sigma}_+$ defined in \eqref{def:sigmaplus} satisfies the following properties:
\begin{enumerate}
\item $\forall \mathcal{N}_1, \mathcal{N}_2 \in \Sigma_+, \, \mathcal{N}_1\mathcal{N}_2 \in \tilde{\Sigma}_+$
\item $\forall F \in \Sigma_0, \quad t \mapsto e^{F(t)} \in \tilde{\Sigma}_+$, (the set $\Sigma_0$ has been introduced in the previous  sequence to \eqref{sigma+}). 
\end{enumerate}
\end{prop}

\begin{proof}
To prove point (1), from the proposition 2.2 of \cite{bergeronetal2012} we have 
that $\mathcal{N}_1$ and $\mathcal{N}_2 \in \Sigma$ implies that $\mathcal{N}_1\mathcal{N}_2 \in \Sigma$. Furthermore $G_{\mathcal{N}_1\mathcal{N}_2, \eta}=G_{\mathcal{N}_1, \eta} G_{\mathcal{N}_2, \eta}$, then if $\mathcal{N}_1$ and $\mathcal{N}_2 \in \tilde{\Sigma}_+$, by definition $G_{\mathcal{N}_1, \eta}$ and $G_{\mathcal{N}_2, \eta} \in \Sigma$. Then from Proposition  2.2 of \cite{bergeronetal2012}, we deduce $G_{\mathcal{N}_1\mathcal{N}_2, \eta} \in \Sigma$.\\
The part (2) is obtained always from the proposition 2.2 of \cite{bergeronetal2012}. First if $F \in \Sigma_0$, then $\mathcal{N}=e^F \in \Sigma$ (proposition 2.2). Furthermore $G_{\mathcal{N}, \eta}(t)=e^{\eta F(t)}$. For $\eta \in ]0,1]$, $t \mapsto \eta F(t)$ also belongs to $\Sigma_0$ then  we deduce always from Proposition 2.2, $G_{\mathcal{N}, \eta} \in \Sigma$.
\end{proof}

\subsection{The characterization of $\tilde{\Sigma}_+$}
In fact point (2) in Proposition \ref{prop:sigmaplus} is not only a sufficient but also a necessary condition to obtain functions of $\tilde{\Sigma}_+$. In order to prove it, we start from the following lemma: 

\begin{lem}
\label{lem:NpN}
For any $\mathcal{N} \in \tilde{\Sigma}_+$, the polynomials $q_n$ verify $q'_n(0) \ge 0$ and furthermore
\begin{equation}
\label{eqn:NpN}
\dfrac{\mathcal{N}'(t)}{\mathcal{N}(t)}= \sum_{n=0}^\infty \dfrac{(n+1) q'_{n+1}(0)}{x_{n+1}!} \, t^n
\end{equation}
\end{lem}
\begin{proof}

To begin with, $q_0(\eta)=1$, so $q'_0(\eta)=0$. For $n \ge 1$, the polynomials $q_n(\eta)$ are nonnegative on the interval $\eta \in (0,1]$ and $q_n(0)=0$. Then $q'_n(0)$ cannot be  negative, otherwise $q_n(\eta)$ would be negative on some interval $(0,\epsilon)$. we conclude that $\forall n \in \mathbb{N}, \, q'_n(0) \ge 0$.\\
Now using the function $G_{\mathcal{N},\eta}(t)=\mathcal{N}(t)^\eta$ and by a differentiation with respect to $t$, we obtain first
\begin{equation}
\forall \eta \in (0,1], \, \mathcal{N}'(t) \mathcal{N}(t)^{(\eta-1)}= \dfrac{1}{\eta} \sum_{n=0}^\infty \dfrac{(n+1)q_{n+1}(\eta)}{x_{n+1}!} t^n
\end{equation}
Taking into account the relation $\forall n \ne 0,  q_n(0)=0$, and assuming that the permutation between $\lim_{\eta \to 0}$ and $\sum$ is valid we obtain the equation \eqref{eqn:NpN}.
\end{proof}

\vspace{0.5cm}

This lemma leads to the theorem

\begin{theo}
\label{theo:LnN}
For all $\mathcal{N} \in \tilde{\Sigma}_+$, $\ln \mathcal{N} \in \Sigma_0$ and
\begin{equation}
\label{eqn:LnN}
\ln \mathcal{N}(t)= \sum_{n=1}^\infty \dfrac{q'_n(0)}{x_n!} \, t^n.
\end{equation}
Furthermore we deduce the characterization
\begin{equation}
\label{characterization1}
\tilde{\Sigma}_+=\{ e^F \, | \, F \in \Sigma_0\}
\end{equation}
\end{theo}

\begin{proof}

Eq.\eqref{eqn:LnN} is immediately obtained by integrating Eq. \eqref{eqn:NpN} term by term and  taking into account the value $\ln \mathcal{N}(0)=0$. 
Using the property $q'_n(0) \ge 0$, which was shown in lemma \ref{lem:NpN}, and the fact that $q'_1(0)= 1$, as $q_1(\eta)=\eta$, we deduce that $\ln \mathcal{N} \in \Sigma_0$.\\ 
The last part of the theorem results from the proposition \eqref{prop:sigmaplus}  (point (2)) and the 
previous comment.
\end{proof}

The theorem \ref{theo:LnN} solves completely the problem of generating functions leading to symmetric deformations of the binomial distribution preserving a statistical interpretation.

\section{Exact solution to the Leibnitz triangle rule}
\label{leibqgauss}
Here we prove a uniqueness result about the Leibnitz triangle rule
\begin{equation}
\label{leibnitzrep}
\varpi_k^{n-1} = \varpi_k^{n} + \varpi_{k+1}^{n} \,, \quad 0\leq k\leq n-1\,. 
\end{equation}
\begin{prop}
\label{leibn1}
With the notations of Section \ref{sec:symmetricbin}, given a sequence of numbers $(x_n)$ with $x_0= 0$, $x_n >0$ for all $n>0$,  with generating function $\mathcal{N}(t) = \sum_{n=0}^{\infty}t^n/x_n!$, and the polynomials $q_n(\eta)$ with generating function $(\mathcal{N}(t))^{\eta}= \sum_{n=0}^{\infty}q_n(\eta)t^n/x_n!$,  the  quantities 
\begin{equation}
\label{tildp1}
 \frac{\mathfrak{p}_k^{(n)}(\eta)}{\binom{n}{k}}= \binom{y_n}{y_k} q_k(\eta) q_{n-k}(1-\eta)
\end{equation}
where $y_n = x_n/n$, are solution to \eqref{leibnitzrep} for $n\geq 1$  if and only if 
the polynomials $q_n$ and the sequence $(y_n)$ read as either
\begin{equation}
\label{solleibnitz1}
q_n(\eta) = \prod_{k=1}^n \frac{ k -1+ \alpha \eta}{k-1 + \alpha}\, , \quad y_n = \frac{\alpha}{n+\alpha -1}\, , \ n\geq 1\, ,
\end{equation}
where $\alpha$ is an arbitrary constant (finite case), or 
\begin{equation}
\label{solleibnitz2}
q_n(\eta) = \eta^n\, , \quad y_n = 1\, \ n\geq 1\, ,
\end{equation}
(infinite case). 
\end{prop}
\Proof
After injection of the expression  \eqref{tildp1} into the two members of the equality \eqref{leibnitzrep} and simplification of the factorials from both sides, we obtain for all $n\geq 1$ and $0\leq k\leq n-1$
\begin{equation}
\label{qqnk}
q_k(\eta) q_{n-k-1}(1-\eta) = \frac{y_n}{y_{n-k}}q_k(\eta) q_{n-k}(1-\eta) + \frac{y_n}{y_{k+1}}q_{k+1}(\eta) q_{n-k-1}(1-\eta)\, .
\end{equation}
We now particularize this identity to the case $k=0$. From Section \ref{sec:symmetricbin}  we know that $q_0(\eta) = 1$ and $q_1(\eta) = \eta$. There follows the recurrence relation
\begin{equation}
\label{qnrec}
q_n(\eta) = \left(1 + \frac{y_n}{y_1}(\eta-1)\right)q_{n-1}(1-\eta) \, .
\end{equation}
and the explicit solution
\begin{equation}
\label{qnexpsol}
q_n(\eta) = \prod_{k=1}^n\left(1 + \frac{y_k}{y_1}(\eta-1)\right) \, .
\end{equation}
We now divide  \eqref{qqnk} by  $q_k(\eta)$ and put $\eta = 0$. From Section \ref{sec:symmetricbin}  we know that $q_n(1) = 1$ for all $n$ and from \eqref{qnrec} we have $q_{k+1}(0)/q_k(0)= 1- y_{k+1}/y_1$. Hence we obtain the functional equation for $Z(n):=1/y_n$:
\begin{equation}
\label{cauchyeq}
Z(n) = Z(n-k) +Z(k+1) - Z(1)\,.
\end{equation}
This equation is of Cauchy-Pexider type \cite{aczel89} and is easily solved by recurrence. Its general solution is $Z(n) = \mu n + \nu$, where $\mu$ and $\nu$ are  arbitrary constants. Now, we note that  the numbers only ratios $y_k/y_1$ appear in the solution \eqref{qnexpsol}. Therefore, for all $k\geq 1$,
\begin{equation}
\label{yny1}
\frac{y_k}{y_1}= \frac{\mu + \nu}{\mu k + \nu}= \frac{1}{\tfrac{\mu}{\mu + \nu}k + \tfrac{\nu}{\mu + \nu}}= \frac{\alpha}{k-1+ \alpha}\, , 
\end{equation}
where we have introduced the parameter $\alpha = 1 + \nu/\mu$. With this expression at hand, we eventually obtain \eqref{solleibnitz1}. The limit case \eqref{solleibnitz2} corresponds to $\alpha \to \infty$ and is trivially verified.
\qed

\section{Proof of the properties of $\Sigma_+^{(1)}$}
\label{proofsigmaprop}

Using Eq. \eqref{qpolybell}, we obtain
\begin{equation}
f_n = \sum_{k=1}^n \dfrac{B_{n,k}(a_1,a_2,\dots,a_{n-k+1})}{B_n(a_1,a_2,\dots,a_n)} k!  \ge \dfrac{n! B_{n,n}(a_1)}{B_n(a_1,a_2,\dots,a_n)}\,.
\end{equation}
Because $a_1=1$ for functions of $\Sigma_+^{(1)}$, and using Eq. \eqref{xqbell}, we find the first property $f_n \ge x_n!$. Therefore we have
\begin{equation}
\forall \eta >0, \, 0 \le \sum_{n=0}^\infty \dfrac{q_n(\eta)}{f_n} \le  \sum_{n=0}^\infty \dfrac{q_n(\eta)}{x_n!} \, ,
\end{equation}
and the right hand side of the inequality is finite because the radius of convergence of $\mathcal{N}(t)^\eta = \sum_{n=0}^\infty \dfrac{q_n(\eta)}{x_n!} t^n$ is greater than 1 for functions of $\Sigma_+^{(1)}$. We conclude that $\sum_{n=0}^\infty \dfrac{q_n(\eta)}{f_n}$ is finite. Finally we have
\begin{equation}
b_{m,n} = \int_0^1 q_m(\eta) q_n(1-\eta) d\eta = \sum_{i=0}^m \sum_{j=0}^n \dfrac{q_m^{(i)}(0) q_n^{(j)}(0)}{(i+j+1)!} \ge \dfrac{q_m^{(m)}(0) q_n^{(n)}(0)}{(m+n+1)!}
\end{equation}
Since for functions belonging to $\Sigma_+^{(1)}$ we have $q_n^{(n)}(0)=x_n!$, we obtain the last inequality $b_{m,n} \ge x_m! x_n! / (m+n+1)!$.

\end{document}